\documentclass{llncs}

\usepackage{makeidx}  % allows for indexgeneration
\begin{document}

\title{Selfish Bin Covering\thanks{The research is supported by the National Natural Science Foundation of China under Grants 70425004, 90818026.}
}

\author{Zhigang Cao, Xiaoguang Yang\\cullencao@gmail.com,  xgyang@iss.ac.cn}
%

%%%% modified list of authors for the TOC (add the affiliations)
\tocauthor{cullen}
\institute{Key Laboratory of Management, Decision \& Information Systems,\\  Academy of Mathematics
and Systems Science,\\ Chinese Academy of Sciences,
Beijing, 100190, P.R. China.\\
\email{}}
\maketitle              % typeset the title of the contribution

\begin{abstract}In this paper, we address the selfish bin covering problem, which is greatly related both to
 the bin covering problem, and to the weighted majority game. What we
mainly concern is how much the lack of coordination harms the social welfare. Besides the standard
PoA and PoS, which are based on Nash equilibrium, we also take into account the strong Nash
equilibrium, and several other new equilibria. For each equilibrium, the corresponding PoA and PoS
are given, and the problems of computing an arbitrary equilibrium, as well as approximating the
best one, are also considered.

\end{abstract}
\section{Introduction}
The bin  covering problem, which is also called  the {\it dual bin packing problem}, since it has
some kind of dual relation with the famous bin packing problem, has $n+1$ positive integers, $a_1,
a_2, \cdots, a_n, b$, as its input, where $a_j$ ($1\leq j\leq n$) is the {\it size} of item $j$ and
$b$ the {\it volume} of bins. We assume that $a_j<b$ for all $1\leq j\leq n$ and there are
sufficiently many bins. The problem is to allocate the items into bins such that the number of {\it
covered} bins  is maximized, where a bin is called covered if the total size of items allocated to
it is equal to or greater than $b$. Alternatively, the problem can be seen as to partition the $n$
items such that the number of subsets in the partition whose total size is not less than $b$ is
maximized.

An implicit assumption  in the bin covering problem is the existence of a central decision maker,
who has absolute control of all the items, and aims to maximize the social welfare, i.e. the number
of covered bins (we can imagine it like this: for each covered bin, a profit of 1 is gained by the
society). This kind of centralized system is often encountered in the real world, and has its own
advantages, say, efficiency. At the same time, centralized systems are usually hard to realize and
really fragile. There is also an opposite kind of system, the decentralized system, where the
central decision maker disappears, and is replaced by a set of  agents, each of which controls a
single item and makes its own decision to maximize its own profit. Compared with the centralized
system, a decentralized system is usually easier to realize and more robust. Due to these
advantages, decentralized systems are drawing more and more attention of researchers in artificial
intelligence, control theory, mechanism design and theoretical computer science.

The biggest disadvantage of decentralized systems, however, is that in most cases, the lack of a
central decision maker will harm the social welfare, that is, it is usually inefficient. Our
question is,  how inefficient?  If the inefficiency is relatively small, then this is no big deal.
So an immediate question is, how to quantify the inefficiency? Fortunately, armed with game theory,
researchers find a satisfactory measure: {\it price of anarchy} (PoA). Game theory studies the
settings where all the players are self-interested, and takes equilibria as the outcomes that may
happen. Since equilibria are usually not unique, PoA is intuitively defined as the ratio between
the worst social welfare we may get among all the equilibria in the decentralized system, and that
in the centralized system (i.e. the optimal social welfare). The analysis of PoA, as well as PoS
(price of stability, which is the ratio between the best social welfare we may get among all the
equilibria and the optimal social welfare), for various decentralized systems, forms an important
part of the booming {\it algorithmic game theory} (\cite{r08}).

We refer to the decentralized bin covering problem as  the  {\it selfish bin covering}  problem
(SBC for short). Before giving the formal definitions of PoA and PoS for SBC, there are still three
questions we have to answer: 1. What are the payoff functions of agents? 2. How to define the
social welfare? 3. How to define an equilibrium?

 For the first question, we take the natural and frequently used proportional rule, that is, in any covered bin,
 the payoff allocated to any member is its size divided by the the total size of items allocated to that bin. For the second,
 we take the popular {\it utilitarian} welfare function, i.e. the sum of all the agents' payoff, which is exactly
  the number of covered bins. For the last,
 we take the famous (pure) Nash equilibrium (NE for short), (pure) strong Nash equilibrium (SNE), and four  newly defined
 ones: fireable Nash equilibria type (I) (FNE(I)), fireable Nash equilibria type (II) (FNE(II)),
fireable Nash equilibria type (III) (FNE(III)), and modified strong Nash equilibrium (M-SNE), whose
exact definitions will be given in the later corresponding sections, respectively.

 Formally, let $SBC$ be the set of all the selfish bin covering problems, $OPT(G)$ the
 optimal social welfare of $G\in SBC$, $NE(G)$ the set of NEs of $G$, and $p(\pi)$
 the social profit of $\pi\in NE(G)$, then PoA and PoS for SBC w.r.t. NE are defined as follows:
 \begin{eqnarray*}PoA^{NE}(SBC)=\inf _{G}\min _{\pi\in NE(G)}\frac{p(\pi)}{OPT(G)},\end{eqnarray*}
 \begin{eqnarray*}PoS^{NE}(SBC)=\inf _{G}\max_{\pi\in NE(G)}\frac{p(\pi)}{OPT(G)}.\end{eqnarray*}

Analogous to $PoA^{NE}$ and $PoS^{NE}$, we can define the other PoAs and PoSs. What we need to do
is simply replace $NE$ in the above two formulas with $SNE$ etc., respectively. In the rest of this
paper, the parameter $SBC$ in various PoAs and PoSs will be omitted.

For each equilibrium, we also concern the problems of computing an arbitrary equilibrium, as well
as approximating the best one (which is NP-hard to find). Our results are listed in the following
table, where complexity stands for the complexity of finding an arbitrary equilibrium, and AFPTAS
(Asymptotic Fully Polynomial Time Approximation Scheme) means that AFPTAS exists, that is, for
arbitrary positive small $\epsilon$, there is an algorithm which outputs an equilibrium whose
social welfare is at least $1-\epsilon$ times the optimal one when the optimal social welfare is
large enough, and the running time can be bounded by a polynomial both in $n$ and $1/ \epsilon$.

\begin{table}[!h]
\tabcolsep 0pt \caption{Main Results} \vspace*{-12pt}
\begin{center}
\def\temptablewidth{1\textwidth}
{\rule{\temptablewidth}{1pt}}
\begin{tabular*}{\temptablewidth}{@{\extracolsep{\fill}}c|cccc}
Equilibria&PoA&PoS&Complexity&Approximability\\
\hline
NE&0&1&$O(1)$&AFPTAS \\
FNE(I)&0.5&1&$O(n^2)$&AFPTAS\\
FNE(II)&0.5&1&NP-hard&NP-hard\\
FNE(III)&0.5&1&NP-hard&NP-hard\\
M-SNE&0.5&1&NP-hard&NP-hard\\
SNE&0.5&0.5&NP-hard&NP-hard\\
\end{tabular*}
       {\rule{\temptablewidth}{1pt}}
       \end{center}
       \end{table}

 \section{Related work}
{\bf The bin covering problem} is first studied by Assman in his thesis (\cite{a83}) and then in
the journal paper of Assmann et al. (\cite{ajkl84}). It is very easy to show that the bin covering
problem is NP-hard, and we are not able to approximate it with a worst case performance ratio
better than 0.5, unless P=NP (\cite{a83}). It is still not hard to show that the next-fit algorithm
has a worst case performance ratio of exactly 0.5 (\cite{ajkl84}). Analogous to the bin packing
problem, what we are interested is how the bin covering problem can be asymptoticly approximated.
To put it precisely, given any instance $I$, let $A(I)$ be the number of bins covered by algorithm
$A$, and $OPT(I)$ be the optimal value. The {\it asymptotic worst case ratio} of algorithm $A$ is
defined as
$$R^A_{\infty}=\liminf\limits_{OPT(I)\rightarrow\infty}\frac{A(I)}{OPT(I)}.$$
One central topic for the bin covering problem is to find a polynomial time algorithm with an
asymptotic worst case ratio as large as possible.

Low-order polynomial time algorithms with $R^A_{\infty}=2/3$ and $3/4$ are explored by Assmann
(\cite{a83}), Assmann et al. (\cite{ajkl84}) and Csirik et al. (\cite{cflz99}). Csirik et al.
(\cite{cjk01}) give the first APTAS (Asymptotic Polynomial Time Approximation Schemes), Jansen and
Solis-Oba (\cite{jo03}) derive an AFPTAS. For the analysis of average case performance, there are
also several results (\cite{cfgk91}, \cite{cjk01}). Woeginger and Zhang (\cite{wz99}) also consider
the variant with variable-sized bins.

{\bf The weighted majority game} (WMG for short), which is a classic model in coalitional
(cooperative) game theory and has rich applications in politics, is greatly related with the
selfish bin covering problem, so we give a sketch of it. An instance of WMG is usually represented
by $G=(q; a_1, a_2, \cdots, a_n)$, where $q$ is called the quota and $a_j$ the weight of player
$j$. $q$ is usually required to be bigger than half of the total weight, and this is the only
difference between WMG and SBC.

The characteristic function is defined as: $v(C)=1$ iff $\sum_{j\in C}a_j\geq q$; $v(C)=0$
otherwise. An interesting phenomenon in WMG is that the bargaining powers of players usually are
not proportional to their weights. For a simple example, there are 3 players in total, their
weights are 3, 4, 5, respectively, and the quota is 7, then all the three players have the same
bargaining power! So a central topic in WMG is how to measure the bargaining powers of the players.
There are mainly four recognized ways: the Shapley-Shubic index (\cite{ss54}), the Banzhaf index
(\cite{b65}), the Holler-Packel index (\cite{h82,hp83}), and the Deegen-Packel index (\cite{dp79}).

Matsui and Matsui (\cite{mm00,mm01}) prove that all the problems of computing the Shapley-Shubik
indices, the Banzhaf indices and the Deegan-Packel indices in WMGs are NP-hard, and there are
pseudo-polynomial time dynamic programming algorithms for them.  Cao and Yang (\cite{cz09}) show
that computing the Holler-Packel index is also NP-hard. Deng and Papadimitriou (\cite{dp94}) prove
that it is \#P-complete to compute the Shapley-Shubik index. Matsui and Matsui (\cite{mm01})
observe that Deng and Papadimitriou's proof can be easily carried over to the problem of computing
the Banzhaf index.

Back to SBC, though the proportional way of payoff allocating in our paper may not be absolutely
fair, the proportion does roughly measure the contribution that an agent makes to the coalition.
What's more, this way  prevails in the real world, and its being trivial to calculate, compared
with the NP-hardness of the more delicate ways , is really an advantage to our analysis.

{\bf Coalition formation games},  which studies which coalition structures (i.e. partitions of
players) are the most likely to form, is a relatively new branch in coalitional game theory. We can
safely classify SBC into coalition formation games. In coalition formation games, {\it stable}
coalition structures are treated as the ones to form. Since the concept of stability is just
another word for equilibrium, which we discussed in the last section, briefly introducing several
closely related stabilities that have already been used in coalition formation games will be
helpful.

Nash stability, i.e. NE that we mentioned in the last section, requires that no player will benefit
if he leaves the original coalition to which he belongs and joins a new one. In this concept, there
is no restriction for the players' migrations. To require that any player's migrations into a
coalition should not harm the original members gives the concept of {\it individual stability} (IE
for short). To further require that any player's migration should not harm members of the original
coalition to which he belongs gives {\it contractually individual stability} (CIE  for short). It's
obvious that NE implies IE, and IE further implies CIE (for more knowledge, please refer
\cite{bj02}). Since we will show that $PoA^{NE}=0$ and $PoS^{NE}=1$, the latter two concepts will
be omitted in our discussion.

{\bf Selfish bin packing} (SBP for short), which combines the idea of decentralization into the
classic bin packing problem, has a great analogy to SBC. In SBP, there are also $n$ items, each of
which has a size and is controlled by a selfish agent, and sufficiently many bins with identical
capacities $b$. The difference is that the total size of items that are packed into a bin should
not exceed $b$, and every nonempty bin incurs a cost of 1, which is shared among its members
proportional to their sizes. This model is introduced by Bil\`{o} (\cite{b06}). The exact
$PoA^{NE}$ is still unknown up to now, and the current best lower bound and upper bound are 1.6416
(by  Epstein and Kleiman (\cite{ek08}), Yu and Zhang (\cite{yz08}), independently) and 1.6428 (by
Epstein and Kleiman (\cite{ek08})), respectively, with a narrow gap to cover. Epstein and Kleiman
also show that $PoS^{NE}=1$, and $1.6067\leq PoA^{SNE}=PoS^{SNE}\leq 1.6210$. Yu and Zhang also
show that computing an NE can be done in $O(n^4)$ time.
\section{Notations and preliminaries}
We denote by $N$ the set of agents, i.e. $N=\{1, 2, \cdots, n\}$. Let $\pi=\{B_1, B_2, \cdots,
B_m\}$ be a partition of $N$, that is $\cup_{i=1}^mB_i=N$, and $B_i\cap B_j=\emptyset$ for all
$1\leq i\neq j\leq m$. We don't distinguish between a bin and the subset of agents in this bin, if
no confusion is incurred. For any subset $B_i$ of $N$, $s(B_i)$ is the total size of the members in
$B_i$, i.e. $s(B_i)=\sum _{j\in B_i}a_j$. A partition $\pi$ is called {\it reasonable} if there is
at most one $B_i$ ($1\leq i\leq m$) such that $s(B_i)<b$.

We use $B_{j,\pi}$ to denote the member in $\pi$ to which agent $j$ belongs. Let $p(\pi)$ be the
social welfare of $\pi$, i.e. the number of covered bins in $\pi$, $p(j, \pi)$ be the payoff
allocated to agent $j$ in $\pi$, then $p(j, \pi)=a_j/s(B_{j,\pi})$ if $s(B_{j,\pi})\geq b$, and
$p(j, \pi)=0$ if $s(B_{j,\pi})<b$. For any $B_i\in \pi$, let $a_{\min}(B_i)=\min\{a_j:j\in B_i\}$.
A bin $B_i$ is called {\it minimal covered} if it is covered and the removal of any of its member
causes it uncovered, i.e. $b\leq s(B_i)<b+a_{\min}(B_i)$, and {\it exactly covered} if $s(B_i)=b$.
$\pi$ is called {\it rational} if it is reasonable and all the covered bins are minimal covered.

The partition problem (\cite{gj79}), which may be the easiest NP-hard problem, is frequently used
in the proving of weakly NP-hardness. It is described as follows: given a set of positive integers
$e_1, e_2, \cdots, e_n$, can we partition these integers into two subsets, such that each of them
has exactly a half of the total size? i.e. is there a set $C\subseteq \{1, 2, \cdots, n\}$ such
that $\sum_{j\in C}e_j=\sum_{j\notin C}e_j$?

The subset sum problem (\cite{gj79}), which will be used later, is also a classic problem. Given a
set of positive integers $e_1, e_2, \cdots, e_n$ and an integer $s$, the problem asks if there
exists a subset which sums to $s$. Subset sum is NP-hard, and admits a pseudo-polynomial time
algorithm with time complexity $O(n\sum_{j=1}^ne_j)$. Subset sum can be seen as a special case of
the knapsack problem, which is known to be linearly solvable if the capacity of the knapsack is a
constant (\cite{p99}). Therefore, if $s$ is constant, subset sum is  also linearly solvable.

$FFD$ (First Fit Decreasing, \cite{j73}), which is a well known heuristic for the bin packing
problem, has an analogue in the bin covering problem. We still denote the analogue  as $FFD$. It
first sorts the items in non-increasing order of their sizes, then puts the items one by one in
this order to a bin until the bin is covered, and opens a new bin to repeat the above action if
there are still unpacked items. Due to the non-increasing order, we can easily see that bins
derived by $FFD$ are all minimal covered, except possibly the last uncovered one.

$LPT$ (Largest Processing Time, \cite{g69}), which is  a well known heuristic for the parallel
machine scheduling problem $P||C_{\max}$, will also be used later. $LPT$ first sorts all the jobs
in non-increasing order of their processing times, and then allocate the jobs one by one in this
order to a least loaded machine. In SBC, the size of an item can be naturally seen as the
processing time, and bins can be seen as machines.  In the next section, we will use $LPT$
partially. That is, a subset of items have already been allocated to bins, and $LPT$ is used to
allocate the remaining items to the already opened bins.

\section{Nash equilibrium}

 According to the  definition, a reasonable partition $\pi=\{B_1,
B_2, \cdots, B_m\}$ is an NE iff $p(j,\pi(j,B_i))\leq p(j,\pi)$ for all $j\in N$ and for all $B_i$
such that $j\notin B_i$, where $\pi(j,B_i)$ is a partition of $N$ derived by moving $j$ to $B_i$.
NE always exist, since it is trivial to see that $\{N\}\in NE$. Remember that we assume $a_j<b$ for
all $j\in N$. And therefore $PoA^{NE}=0$.

From any reasonable partition $\pi=\{B_1, B_2, \cdots, B_m\}$, we show that it is easy to derive an
NE, without decreasing the social welfare. Suppose that the unique uncovered bin in $\pi$ is $B_m$.

 \begin{center}Algorithm $FFD$-$LPT$\end{center}

{\it STEP 1.} Input the reasonable partition $\pi=\{B_1, B_2, \cdots, B_m\}$;

 ~~~~~~~~~~~~{\bf for} $1\leq i\leq m-1$

 ~~~~~~~~~~~~~~~~{\bf while} ($B_i$  is not minimal covered)

 ~~~~~~~~~~~~~~~~~~~~~~~~\{take out the smallest remaining
 item;\}

 ~~~~~~~~~~~~{\bf endfor}

{\it STEP 2.} Run $FFD$ on the taking out items and the items in $B_m$;

~~~~~~~~~~~~/*Without loss of generality, we assume there are still $m-1$ covered bins in total,
and the unique uncovered bin is $B_m$.*/

{\it STEP 3.} {\bf while} ($\max\{a_j: j\in B_m\}>\min\{a_j: j\in \cup_{i=1}^{m-1}B_i\}$)

~~~~~~~~~~~~~~~~~~~~~\{exchange the biggest item in $B_m$ with the smallest one in
$\cup_{i=1}^{m-1}B_i$;

~~~~~~~~~~/*Suppose $B_k$ is the bin which consists of the smallest item in $\cup_{i=1}^{m-1}B_i$
before the exchange.*/

 ~~~~~~~~~~~~~~~~~~{\bf while} ($B_k$  is not minimal covered)

 ~~~~~~~~~~~~~~~~~~~~~~~~~~ \{take out the smallest remaining
 item;\}

 ~~~~~~~~~~~~~~~~~~Run $FFD$ on the taking out items and items in $B_m$.\}

{\it STEP 4.} Call $LPT$ to allocate the items of $B_m$ into $B_1, B_2, \cdots, B_{m-1}$.

\begin{lemma}For any reasonable partition $\pi$,  $FFD$-$LPT$ gives an NE in
 $O(n^2)$ time, without decreasing the social welfare.\end{lemma}
\begin{proof}Let $\pi'$ be the final partition. It is trivial that $p(\pi')\geq p(\pi)$.
Without loss of generality, we suppose $\pi'=\{B_1', B_2', \cdots, B_{m-1}'\}$.
  To verify that $\pi'\in NE$, it suffices to show that for all $1\leq i\leq m-1$, the smallest agent
in $B_i'$ is satisfied with its position.

We denote the partition at the end of Step 3 still as $\pi=\{B_1, B_2, \cdots, B_m\}$, and have the
next three observations: (a) each item in $\cup_{i=1}^{m-1}B_i$ is not bigger than any one in
$B_m$; (b) each $B_i$ ($1\leq i\leq m-1$) is minimal covered; (c) every agent in
$\cup_{i=1}^{m-1}B_i$ is satisfied with its position in the partial partition $\{B_1, B_2, \cdots,
B_{m-1}\}$. (a) is just an interpretation of the stopping condition of Step 4, (b) is due to $FFD$,
and (c) comes directly form (b).

For any $1\leq i\leq m-1$ and $j\in B_i'\in \pi'$ with $a_j=a_{\min}(B_i')$, we discuss in two
cases. {\bf Case 1.} $j$ is in $B_m$ at the end of Step 3. Due to $LPT$, $B_i$ has the smallest
total size before the joining of $j$, and therefore $j$ is temporarily  satisfied with its
position. $a_j=a_{\min}(B_i)$ tells us that $j$ is  the last member of $B_i$. After the joining of
$j$, the total size of $B_i$ keeps unchanged, while the total sizes of the other bins do not
decrease, so $j$ is satisfied with its position in the final partition $\pi'$. {\bf Case 2.} $j$ is
in $B_i$, $1\leq i\leq m-1$,  at the end of Step 3. Due to observation (a), there is no member
joining $B_i$ in Step 4. Since agent $j$ is satisfied at the end of Step 3 (observation (c)), it is
satisfied in the final partition $\pi'$.

Step 1 and Step 2 can be done in $O(n\log n)$ time. Since no exchanged item in $B_m$ comes back,
the outer while loop in Step 3 is executed for at most $n$ times. And it is easy to see that each
loop can be done in $O(n)$ time. So the total time for Step 3 is $O(n^2)$. Step 4 can be  easily
done in $O(n\log n)$ time. In sum, the total running time is $O(n^2)$. \qed\end{proof}

\begin{theorem}(a) $PoA^{NE}=0, PoS^{NE}=1$;

(b) There exists an AFPTAS for approximating the best NE.\end{theorem}

\begin{proof}Since there is
always an optimal partition that is reasonable, we immediate have (a) from Lemma 1. For any
partition derived by the AFPTAS for the bin covering problem (\cite{jo03}), we can also give a
corresponding NE by {\it FFD-LPT} without decreasing the social welfare, so (b) is
valid.\qed\end{proof}

\section{Fireable Nash equilibria}

Let  $\pi=(B_1, B_2,\cdots, B_{m})$ be a reasonable partition, and $B_m$ the unique uncovered bin,
we define three fireable Nash equilibria. The same idea in these equilibria is that, every covered
bin should be minimal covered. This idea is quite natural, since any non-minimal covered bin is
redundant, firing certain agent or agents will make it minimal covered, and the other agents have
incentive to do this unanimously because they will all become strictly better off.

The differences lie in how to constrain the migrations of agents. Recall that in NE, there is no
constraint for the agents' migrations: as long as an agent can benefit from doing so, it is allowed
to do it; in IE, the migration of an agent should not harm its new partners; and in CIE, neither
should the migration harm its new partners nor should it harm the old colleagues.
\subsection{Definitions}
\begin{definition} $\pi\in FNE(I)$  iff:

(1) $\pi$ is rational.

(2) There exists no item $j\in\cup_{i=1}^{m-1}B_i$ such that (a) $B_m\cup \{j\}$ is minimal
covered, and (b) $s(B_m)+a_j<s(B_{j,\pi})$.
\end{definition}

In FNE(I), we require that migration of any agent from a minimal covered bin into the uncovered bin
$B_m$ be not allowed, if it causes $B_m$ non-minimal covered, even if this migration is applauded
by the migrating agent (and of course by the members in $B_m$).

\begin{definition}$\pi\in FNE(II)$ iff:

(1) $\pi$ is rational.

(2) There exists no item $j\in\cup_{i=1}^{m-1}B_i$ and a subset $E$ of $B_m$ such that (a)
$(B_m\setminus E)\cup \{j\}$ is minimal covered, and (b) $s(B_m\setminus E)+a_j<s(B_{j,\pi})$.
\end{definition}

In FNE(II), the migration that we remarked for FNE(I) is allowed. And the migrating agent is
actually farsighted: if it will benefit from the migration only after the  firing of some subset of
the original members, which it anticipates will happen,  then it will do it. Note that the
migration of any agent, either from a minimal covered bin or from the uncovered bin, into a minimal
covered bin, is not allowed, even if the migrating agent anticipates that this will benefit it
eventually. Further allowing this kind of migration gives FNE(III).

\begin{definition}$\pi\in FNE(III)$  iff:

(1) $\pi$ is rational.

(2) There exists no item $j\in N$, a bin $B_i$ ($1\leq i\leq m, B_i\neq B_{j,\pi}$) and a subset
$E$ of $B_i$ such that (a) $(B_i\setminus E)\cup \{j\}$ is minimal covered, and (b) $s(B_i\setminus
E)+a_j<\min\{\delta(B_i), \delta(B_{j,\pi})\}$, where $\delta(B)=s(B)$ if $s(B)\geq b$, and
$\delta(B)=+\infty$ otherwise, for any $B\subseteq N$.
\end{definition}

In the above definition, it is easy to check that (b) means if the new bin $(B_i\setminus E)\cup
\{j\}$ forms, then all its members will be strictly better off than in $\pi$. It is also easy to
see that FNE(III) implies FNE(II), which further implies FNE(I), since the constraints for agents'
migrations are the most in FNE(I), and are the least in FNE(III).

\subsection{The potential function}
We define a potential function which will be used in the next subsection and Section 7. Unlike the
ordinary single valued potential function, our potential function is vector valued. Let $P(\cdot)$
be the potential function, and $\pi=(B_1, B_2,\cdots, B_{m})$ is a reasonable partition, and $B_m$
the unique uncovered bin, then $$P(\pi)=(s(B_1'),s(B_2'), \cdots, s(B_{m-1}')),$$ where $\{B_1',
B_2', \cdots, B_{m-1}'\}=\{B_1, B_2, \cdots, B_{m-1}\}$ and $s(B_1')\leq s(B_2')\leq \cdots \leq
s(B_{m-1}')$.

For any two vectors $v=(v_{1},v_{2}, \cdots, v_{s}), w=(w_1,w_2,\cdots, w_t)$, we say that $v$ is
lexicographically smaller than $w$, which is denoted by $v\prec w$, iff there exists an integer
$j_0$, $1\leq j_0\leq \min\{s,t\}$, such that $v_{j_0}<w_{j_0}$, and $\forall j<j_0$, $v_j=w_j$.

Accordingly, we can define an order $\prec$ between reasonable partitions. Let $\pi^1,\pi^2$ be two
reasonable partitions, we say that $\pi^1\prec \pi^2$ iff $P(\pi^1)\prec P(\pi^2)$. It's easy to
check that this is a strictly partial order, that is, it satisfies ir-reflexivity ($\neg (\pi \prec
\pi$)), asymmetry ($\pi^1\prec \pi^2\Rightarrow \neg (\pi^2\prec \pi^1$)) and transitivity
($\pi^1\prec \pi^2$ and $\pi^2\prec \pi^3$ $\Rightarrow \pi^1\prec \pi^3$).

In the setting of FNE(I), where the migration can only occur from a covered bin to the uncovered
one, the natural best response dynamics, in which an arbitrary unsatisfied agent moves to the
uncovered bin unless the partition is already an equilibrium, always lead to an equilibrium, since
each migration makes the previous partition smaller in the sense of $\prec$, and  the number of
reasonable partitions is finite for any fixed instance of SBC. Due to this argument, we get
immediately that FNE(I) exists. Our next question is, is the best-response-dynamics algorithm
polynomial? The answer is YES.
\subsection{Main results}
 The basic ideas for the next two lemmas are from G. Yu and G. Zhang(\cite{yz08}). Starting from a rational
  partition $\pi=(B_1, B_2,\cdots, B_{m})$, where $s(B_1)\geq s(B_2)\geq \cdots\geq s(B_{m})$, and $B_m$ is the
unique uncovered bin, suppose the best-response-dynamics algorithm has $K$ migrations before it
ends. Let the partition after the $t$-th ($1\leq t\leq K$) migration be $\pi^t=(B_1^t,B_2^t,\cdots,
B_m^t)$, and suppose that these bins have been re-indexed such that $s(B^t_1)\geq s(B^t_2)\geq
\cdots \geq s(B^t_m)$. Then:

\begin{lemma}$s(B_m)<s(B_m^1)<s(B_m^2)<\cdots<s(B_m^K)$, and $s(B_j)\geq s(B^1_j)\geq s(B_j^2)\geq \cdots \geq s(B_j^K)$ for all $1\leq j\leq m-1$.\end{lemma}
\begin{proof}Since in any step the minimal covered bin to which the migrating agent belongs becomes the next unique uncovered bin, the first
part of the lemma is trivial. Because the total size of the new minimal covered bin is not smaller
than the one to which the migrating agent originally belongs, the second part is also easy.
\qed\end{proof}
\begin{lemma}Suppose agent $i\in N$ migrates $t_i$ times in the best-response-dynamics algorithm, then $t_i\leq m-1$. And hence $K\leq n(m-1)$.\end{lemma}
\begin{proof}Suppose the bin that agent $i$ moves in for the $j$-th time ranks $r_j$ in the derived partition, $1\leq j\leq t_i$.
For any fixed $j$, $j\geq 2$, let $\pi^x, \pi^y$ and $\pi^z$ be the partition after the $j-1$-th
move of agent $i$, the partition before the $j$-th move of agent $i$, and the partition after the
$j$-th move of agent $i$, respectively. It is trivial that $x<y<z$. We also have $i\in
B^x_{r_{j-1}}$, $i\in B^z_{r_j}$, $s(B^z_{r_j})=s(B_m^{z-1})+a_i$ and
$s(B_{r_{j-1}}^x)=s(B_m^{x-1})+a_i$. Due to the first part of Lemma 2, we have
$s(B^{z-1}_{m})>s(B_{m}^{x-1})$ and therefore $s(B^z_{r_j})>s(B_{r_{j-1}}^x)$. Due to the second
part of Lemma 2, we have $s(B^z_{r_{j-1}})\leq s(B^x_{r_{j-1}})$. Therefore,
$s(B_{r_j}^z)>s(B^z_{r_{j-1}})$, and so $r_j<r_{j-1}$. This completes the proof.\qed\end{proof}

 \begin{theorem}(a) FNE(I) always exists, and computing an FNE(I) can be done in $O(n^2)$ time;

 (b) $PoS^{FNE(I)}=1$;

 (c) There exists an AFPTAS for approximating the best FNE(I).\end{theorem}
 \begin{proof}Since $m<n$, and a rational partition is easy to calculate (say, by $FFD$), in $O(n)$ time,
 (a) is straightforward. (b) and (c) are valid for the same reasons as in Theorem 1. \qed\end{proof}

 \begin{theorem}$PoA^{FNE(I)}=PoA^{FNE(II)}=$ $PoA^{FNE(III)}=0.5$.\end{theorem}
\begin{proof}The minimal covered property guarantees that none of the three equilibria has a $PoA$ less than 0.5.
To prove  that they are all exactly 0.5, we only have to show that $PoA^{FNE(III)}=0.5$, since
$FNE(III)\subseteq FNE(II)\subseteq FNE(I)$.

We construct an instance as follows. There are $6n$ items in total: $2n$ large items with
$a_1=a_2=\cdots=a_{2n}=2n-2,$ $4n$ small ones with $a_{2n+1}=a_{2n+2}=\cdots a_{6n}=1$, and the
volume of the bins is $b=2n$. It is easy to check that the partition that each bin contains either
two large items or $2n$ small ones is an FNE(III). And the partition each bin contains one large
item and two small ones is an optimal partition. Since $(n+2)/(2n)\rightarrow 0.5$ as $n\rightarrow
\infty$, we complete the whole proof.\qed
\end{proof}

 \begin{lemma}For any rational partition $\pi$,
 there exists a $\pi'\in FNE(III)$, such that $p(\pi')\geq p(\pi)$.\end{lemma}
\begin{proof}
If $\pi$ is not an FNE(III), then there exists an agent $j\in N$ and a subset $C_j$ of $B_i$
($1\leq i\leq m, B_i\neq B_{j,\pi}$) such that $b\leq a_j+s(B_i\setminus C_j)<s(B_{j,\pi})$.  We
can construct a new rational partition $\pi'\prec\pi$,  without decreasing the social welfare, in
at most five steps as follows.

 1.  Move item $j$ into $B_i$;

 2. If $C_j\neq \emptyset$, move the items in $C_j$ into $B_{j,\pi}$;

 3. If $(B_{j,\pi}\setminus \{j\}\cup C_j)$ is non-minimal covered, take out some subset
 $D_j$ of items such that it is minimal covered;

 4 .If all the bins are covered, open a new bin with items in $D_j$, else put the items in $D_j$
into the unique uncovered bin, which  must be $B_m$;

5. If $B_m\cup D_j$ is non-minimal covered, take out a subset $E_j$ of items such that it is
minimal covered, and open a new bin with items in $E_j$.

It is not hard to check that $\pi'\prec\pi$ and $p(\pi')\geq p(\pi)$. This completes the proof.\qed
\end{proof}

\begin{theorem}$PoS^{FNE(II)}=$ $PoS^{FNE(III)}=1$.\end{theorem}
\begin{proof} Lemma 4 tells us that $PoS^{FNE(III)}=1$. Since FNE(III) is stronger than FNE(II), the theorem is given.\qed\end{proof}
\begin{theorem}Computing an arbitrary FNE(II) is NP-hard, and so is computing an arbitrary FNE(III).\end{theorem}
\begin{proof}We prove this by reduction to the partition problem. Given any instance of the partition problem: $e_1, e_2, \cdots, e_n$,
we construct an instance of SBC as follows.

There are $n+2$ items in total: $n$ small ones $a_1=2e_1, a_2=2e_2,\cdots, a_n=2e_n,$ and 2 large
ones $a_{n+1}=a_{n+2}=\sum_{j=1}^{n}e_j+1,$ and the volume of the bins is $b=\sum_{j=1}^{n}2e_j+1$.
Let $\pi$ be an FNE(II). It suffices to show that $\pi$ has a bin that is exactly covered iff the
answer to the partition problem is yes.

The necessity part is trivial. We show the sufficiency part. Let $C$ be a set of items, and
$\sum_{j\in C}e_j=(1/2)\sum_{j=1}^{n}e_j$. Since the total size of all the items is exactly $2b$,
we know that $\pi$ has exactly two bins. Suppose that $\pi=\{B_1,B_2\}$, and $B_1$ is minimal
covered. There are two possibilities: {\bf case 1.} $B_1$ has two large items; {\bf case 2.} $B_1$
has one large item and a set of small ones. In case 1, $B_2$ will be exactly the set of small
items, and it is uncovered. Therefore, item $n+1$ will be strictly better off if he moves to $B_2$
and pushes out the items in $C$. This contradicts the fact that $\pi$ is an FNE(II). Therefore,
this case will not occur.

In case 2, suppose $B_1$ is made up of item $n+1$ and a set $D$ of small items. We claim that $B_1$
is exactly covered, which will complete the whole proof. If not, we will have $B_2$ is uncovered
and  $\sum_{j\notin D}e_j<(1/2)\sum_{j=1}^{n}e_j$. Then $\sum_{j\in D}e_j\geq
(1/2)\sum_{j=1}^{n}e_j+1$. Therefore $s(B_1)=a_{n+1}+\sum_{j\in D}a_j\geq b+3>b+2$, and item $n+1$
will be strictly better off if it moves to $B_2$ and pushes out all the items other than item
$n+2$, because $a_{n+1}+a_{n+2}=b+2$. This is a contradiction.\qed
\end{proof}

\section{Strong Nash equilibrium}

A partition $\pi=(B_1, B_2, \cdots, B_m)$ is an SNE iff no group of agents can become all strictly
better off by forming a new bin. That is, there is no subset $B\subseteq N$ such that $s(B)\geq b$
and $s(B)<s(B_{j,\pi})$ for all $j\in B$. It is trivial to see that SNE is stronger than NE and
FNE(III).

 For any set of items $E$ and its subset $F$, if
$s(F)=\min\{s(G):G\subseteq E, s(G)\geq b\}$, we say that $F$ is a minimum subset w.r.t. $(E,b)$.
The next lemma characterizes the SNE.
\begin{lemma}Suppose $\pi=(B_1, B_2, \cdots, B_m)$ is a reasonable partition, $B_m$ is the unique uncovered bin, and
 $s(B_1)\leq s(B_2)\leq \cdots\leq s(B_{m-1})$. Then $\pi$ is an SNE iff $B_k$ is a minimum subset w.r.t.
 $(\cup_{i=k}^mB_i,b)$ for all $1\leq k\leq m-1$. \end{lemma}
 \begin{proof}(sufficiency) Suppose that $\pi$ is not an SNE, then there is a subset $B\subseteq N$ such that
$s(B)\geq b$ and $s(B)<s(B_{j,\pi})$ for all $j\in B$. Let $k_0$ be the smallest index such that
$B\cap B_{k_0}\neq \emptyset$, and $j_0\in B\cap B_{k_0}$. Then $B\subseteq \cup_{i=k_0}^mB_i$. By
hypothesis,  $s(B_{k_0})=\min\{s(G):G\subseteq \cup_{i=k_0}^mB_i, s(G)\geq b\}$. So $s(B_{k_0})\leq
s(B)$. This contradicts $s(B)<s(B_{k_0})$.

(necessity) Suppose that $B_{k_0}$ is not a minimum subset w.r.t.
 $(\cup_{i={k_0}}^mB_i,b)$, $1\leq k_0\leq m-1$, that is $s(B_{k_0})>\min\{s(G):G\subseteq \cup_{i=k_0}^mB_i, s(G)\geq
b\}$. Let $B$ be a minimum subset w.r.t.
 $(\cup_{i={k_0}}^mB_i,b)$, then $b\leq s(B)<s(B_{k_0})$. Since $B\subseteq \cup_{i={k_0}}^mB_i$, we know $s(B_{j,\pi})\geq s(B_{k_0})>s(B)$
 for all $j\in B$. Therefore, $s(B)<s(B_{j,\pi})$ for all $j\in B$. A contradiction with $\pi\in SNE$. \qed\end{proof}
\begin{theorem}SNE exists, and to compute an arbitrary SNE is weakly NP-hard.\end{theorem}
\begin{proof}We only have to show the second part. We prove it by reduction to the partition problem. Let $e_1, e_2,\cdots, e_n$ be the input of the partition problem,
and an instance of SBC is constructed simply as $a_j=e_j, 1\leq j\leq n$, and
$b=(1/2)\sum_{j=1}^{n}e_j$. Let $\pi$ be an SNE of SBC. It is straightforward  that $\pi$ has two
exactly covered bins iff the answer to the partition problem is yes.

To prove that the  problem is weakly NP-hard, it suffices to show that a pseudo-polynomial time
algorithm is admitted. This is obvious since the calculating of $s(F)=\min\{s(G):G\subseteq E,
s(G)\geq b\}$ can be done in pseudo-polynomial time simply by recursively calling the algorithm to
the subset sum problem. \qed\end{proof}

The next two corollaries are both obvious.
\begin{corollary}An arbitrary SNE, which is also FNE(III) and FNE(II), can be computed in $O(n\sum_{j=1}^ne_j)$ time.\end{corollary}
\begin{corollary}If $b$ is a constant, then all the equilibria discussed in our paper can be computed in $O(n^2)$ time.\end{corollary}
\begin{theorem}$PoA^{SNE}=PoS^{SNE}=0.5$.\end{theorem}
\begin{proof}For any $\pi=(B_1,B_2,\cdots, B_m)\in SNE$, where $B_m$ is the unique uncovered bin, we have $s(B_i)<2b$ for
all $1\leq i\leq m-1$, then $PoA^{SNE}\geq 0.5$. It suffices to show that $PoS^{SNE}\leq 0.5$.

For any $n$, we  construct an instance as follows. There are $4n$ items in total: $2n$ large items
$a_1=a_2=\cdots=a_{2n}=2n-1$, and $2n$ small ones $a_{2n+1}=a_{2n+2}=\cdots=a_{4n}=2$, and $b=2n$.
Then all SNEs are made up of two small bins, each of which contains $n$ small items,  and $n$ large
bins, each of which contains two large items. And all of them have a  social profit of $n+2$.
However, if we pair each large item with a small one, we will arrive at a social profit of $2n$.
Since $(n+2)/(2n)\rightarrow 0.5$ ($n\rightarrow \infty$), we get the whole theorem.\qed\end{proof}

\section{Modified strong Nash equilibrium}
In the definition of FNE(I), further allowing the simultaneous migration of a group of agents from
the same minimal covered bin into the unique uncovered bin gives the concept of modified strong
Nash equilibrium (M-SNE). It is easy to see that M-SNE is  stronger than FNE(I), but weaker than
SNE. The exact definition is given as follows. Suppose that $\pi=(B_1, B_2, \cdots, B_m)$ is a
reasonable partition, and $B_m$ is the unique uncovered bin.
\begin{definition}$\pi$ is said to be an M-SNE, iff it is rational, and for any $1\leq i\leq m-1$, there is no subset $E\subseteq B_i$, such that
$s(B_i\setminus E)>s(B_m)$ and $B_m\cup E$ is minimal covered.\end{definition}
\begin{lemma}For any rational partition $\pi$, there exists an $\pi'\in M-SNE$, such that $p(\pi')\geq p(\pi)$.\end{lemma}
\begin{proof}If $\pi$ is not an M-SNE,  we can construct a partition $\pi'$ from $\pi$ by letting an unsatisfied group of agents migrating into
the unique uncovered bin. It is easy to see that $\pi'\prec \pi$. \qed\end{proof}
\begin{theorem}$PoA^{M\mbox{-}SNE}=0.5,PoS^{M\mbox{-}SNE}=1$.\end{theorem}
\begin{proof}Lemma 6 guarantees that $PoS^{M\mbox{-}SNE}=1$. $PoA^{SNE}=0.5$ tells us that $PoA^{M\mbox{-}SNE}\leq 0.5$.
 Since for each bin $B$ in a rational partition,
$s(B)<2b$, we get $PoA^{M\mbox{-}SNE}\geq 0.5$. Therefore, $PoA^{M\mbox{-}SNE}=0.5$.
\qed\end{proof}
\begin{theorem}To compute an arbitrary M-SNE is NP-hard.\end{theorem}
\begin{proof}We still prove by reduction to the partition problem. Given any instance of the partition problem: $e_1, e_2, \cdots, e_n$, we construct an instance
of SBC as follows.

There are $n+1$ items in total: $a_j=2e_j$, for all $1\leq j\leq n$, $a_{n+1}=\sum_{j=1}^ne_j-1$,
and $b=\sum_{j=1}^n2e_j-1$. Let $\pi=(B_1,B_2)$ be an M-SNE. Obviously, exactly one of the two bins
is covered. Suppose $B_1$ is covered. We prove that the answer to the partition problem is yes iff
$B_1$ is exactly covered.

The sufficiency part is trivial, so we only have to show the necessity part. Suppose $\sum_{j\in
C}e_j=(1/2)\sum_{j=1}^ne_j$. There are two possibilities: $n+1\in B_1$ and $n+1\notin B_1$. In the
latter case, we have $B_1=\{1,2,\cdots, n\}$. So $s(B_1)=\sum_{j=1}^n2e_j=b+1>b$, the  items in $C$
will benefit if they move into $B_2$. Therefore this case will not occur. In the former  case, if
$s(B_1)>b$, we will have $s(B_1)\geq b+2$, so the items in $B_1\setminus \{n+1\}$ will benefit if
they move into $B_2$, a contradiction. \qed\end{proof}

Since M-SNE is defined analogously to FNE(I), it should actually be written as M-SNE(I), and we can
also define M-SNE(II) and M-SNE(III), analogously to FNE(II) and FNE(III), respectively. It is not
hard to show that all the results for M-SNE, i.e.  NP-hardness, PoA=0.5 and PoS=1,  also hold for
M-SNE(II) and M-SNE(III).

\section{Further discussions}In this paper, we discussed the selfish bin covering problem. In order to see how decentralized decision making affects
the social profit, various PoAs and PoSs are analyzed. We mainly considered
 six equilibrium concepts, whose relations are summarized as follows.

\begin{center}
$FNE(III)\Rightarrow FNE(II)\Rightarrow FNE(I)$\\
$~~\Uparrow~~~~~~~~~~~~~~~~~~~~~~~~~~~~~~~~\Uparrow$
\\$~~~~SNE~~~~~~~~~~\Rightarrow~~~~~~~~~~~~ M$-$SNE$\\
$\Downarrow\hskip 3.9cm$\\
$NE\hskip 4cm$\end{center}

It is easy to show that the above relation graph is complete, i.e., all the other relations that
are not expressed in this graph are false.

It's not hard to check that if we consider asymptotic PoAs and asymptotic PoSs, i.e. in the
definitions we change inf as $\liminf$ and let $OPT(G)\rightarrow \infty$, then all the results
remain the same. This can be interpreted as that these measures are pretty stable with regard to
the size of the problems.

We remark that the algorithm FFD-LPT in Section 4 is not a best-response-dynamic algorithm. An open
problem is that starting from any rational partition, does any best-response-dynamic algorithm,
which always converges to an NE, ends in polynomial time, as we showed for FNE(I)? We know this is
impossible for FNE(II) or FNE(III), as they are both NP-hard to compute. It's also interesting to
study whether SNE and M-SNE can be arrived at in polynomial number of migrations.

It's a pity that we haven't  found the applications of selfish bin covering yet, but we believe
that it deserves study just because of its simplicity.  The following are some interesting
directions for further research.

1. To  consider and design other payoff allocation rules, e.g. the membership rule, i.e. the payoff
of a covered bin is shared equally among its members, allocation rule proportional to varieties of
power indices, etc.;

2. To consider the {\it egalitarian} social welfare function, i.e. the social welfare is determined
by the lowest payoff among all the players. The corresponding combinatorial optimization problem
can be easily shown to be strongly NP-hard, by reduction from the 3-partition problem;

3. To take a hybrid social welfare of the utilitarian rule and the egalitarian rule, either
aggregating them into one function using the weighting way, or discussing them lexicographically,
or any other method dealing with bi-criteria problems;

4. To extend the model to heterogenous bin payoffs, i.e. the payoff gained by every covered bin is
a function of its total size. Notice that this is closely related to the congestion game with a
parallel network;

5. To extend the model to heterogenous bins: variable sized bins, variable payoff  bins. Notice
that these models are closely related with the selfish load balancing, i.e. scheduling games, on
uniformly related machines;

6. To extend the model to higher dimensional cases, i.e. items and bins are characterized by
vectors of a fixed dimension;

7. To discuss the ordinal model, i.e. how much a covered bin gains is determined by its rank among
all the covered bins. Notice that externality occurs in this model;

8. To discuss incompatible agent families, i.e. agents from the same family can not be packed in
the same bin;

9. To discuss a hybrid model of centralization and decentralization, i.e. each agent may control
more than one items;

10. To discuss online coalition formations, i.e. agents arrive one by one, and each agent decides
either to join an old bin, or to open a new bin, e.g. the procedure proposed by E. Maskin
(\cite{m03}).

 {\bf ACKNOWLEDGEMENTS.} The authors would like to thank Prof. G.C. Zhang for the wonderful lecture he
delivered about selfish bin packing, which gives the basic motive of this paper. They are also
grateful to Prof. M.C. Cai for many helpful discussions.

\end{document}